\documentclass{article}
\usepackage{amsmath}
\usepackage{amsfonts}
\usepackage{tabularx}
\usepackage{enumerate}

\newtheorem{theorem}{Theorem}

\newtheorem{corollary}[theorem]{Corollary}

\newtheorem{definition}[theorem]{Definition}
\newtheorem{example}[theorem]{Example}

\newenvironment{proof}[1][Proof]{\noindent\textbf{#1.} }{\ \rule{0.5em}{0.5em}}

\newcommand{\add}{{\mathbb{Z}_2 {\mathbb{Z}_4}}}
\newcommand{\ad}{{\mathbb{Z}_2 {\mathbb{Z}_{4}[\xi]}}}
\newcommand{\ki}{{\mathbb{Z}_{2}^r[\bar{\xi}]\times {\mathbb{Z}_{4}^s[\xi]}}}
\newcommand{\ZZ}{{\mathbb{Z}}}

\newcommand{\C}{{\mathcal{C}}}

\newcommand{\type}{{\left(r,s;k_{0};k_{1},k_{2}\right)}}

\begin{document}

\title{On $\mathbb{Z}_{2}\mathbb{Z}_{4}[\xi]$-Skew Cyclic Codes}

\author{Ismail Aydogdu$^{a}$\thanks{{\footnotesize iaydogdu@yildiz.edu.tr
(Ismail Aydogdu)}}, Fatmanur Gursoy$^{a}$\thanks{{\footnotesize fatmanur@yildiz.edu.tr(Fatmanur Gursoy)}} \\
$^{a}${\footnotesize Department of Mathematics, Yildiz Technical University}
}

\maketitle

\begin{abstract}
$\mathbb{Z}_2\mathbb{Z}_{4}$-additive codes have been defined as a subgroup of $\mathbb{Z}_2^{r}\times \mathbb{Z}_4^{s}$ in \cite{3} where $\mathbb{Z}_2$, $\mathbb{Z}_{4}$ are the rings of integers modulo $2$ and $4$ respectively and $r$ and $s$ positive integers. In this study, we define a new family of codes over the set $\mathbb{Z}_2^{r}[\bar{\xi}]\times \mathbb{Z}_4^{s}[\xi]$ where $\xi$ is the root of a monic basic primitive polynomial in $\mathbb{Z}_{4}[x]$. We give the standard form of the generator and parity-check matrices of codes over $\mathbb{Z}_2^{r}[\bar{\xi}]\times \mathbb{Z}_4^{s}[\xi]$ and also we introduce skew cyclic codes and their spanning sets over this set.
\end{abstract}

\begin{quotation}
\end{quotation}

\bigskip
\section{Introduction}

The structure of binary linear codes and quaternary linear codes have been
studied in details for the last sixty years. Recently, a new class of error
correcting codes over the ring $\mathbb{Z}_{2}\times\mathbb{Z}_{4}$ called additive codes that
generalizes the class of binary linear codes and the class of quaternary
linear codes has been studied in \cite{3}. A $\mathbb{Z}_{2}\mathbb{Z}_{4}$-additive code $\mathcal{C}$ is defined
to be a subgroup of $\mathbb{Z}_{2}^{r}\times\mathbb{Z}_{4}^{s}$ where $
r+2s=n$. If $s=0$ then $\mathbb{Z}_{2}\mathbb{Z}_{4}$- additive codes are just
binary linear codes, and if $r=0,$ then $\mathbb{Z}_{2}\mathbb{Z}_{4}$-additive codes
are the quaternary linear codes over $\mathbb{Z}_{4}$.

Lately, Aydogdu et al.(see \cite{1} and \cite{6}) have generalized $\mathbb{Z}_{2}\mathbb{Z}_{4}$-additive codes to $\mathbb{Z}_2\mathbb{Z}_{2^s}$ and then $\mathbb{Z}_{p^r}\mathbb{Z}_{p^s}$ additive codes where $p$ is a prime and $r,s$ are positive integers. In 2014, $\mathbb{Z}_2\mathbb{Z}_{2}[u]$-additive codes were introduced in \cite{2}, which are actually a different kind of generalization of $\mathbb{Z}_{2}\mathbb{Z}_{4}$-additive codes. And also at the same year, Abualrub et al. presented a paper ``$\mathbb{Z}_{2}\mathbb{Z}_{4}$-Additive Cyclic Codes" that they introduced cyclic codes over $\add$ and gave the spanning sets of these cyclic codes \cite{7}.

Skew cyclic codes was first introduced by Boucher et al. in \cite{D.B}. They generalized cyclic codes to a class of linear codes over skew polynomial rings with an automorphism $\theta$ over the finite field with $q$ elements $(\mathbb{F}_q)$. This polynomial ring is a set of ordinary  polynomials denoted by; $\mathbb{F}_q[x;\theta]=\{a_0+a_1x+\ldots+a_kx^k|a_i\in \mathbb{F}_q,0\leq i\leq k\}$, where addition is defined as the usual polynomial addition and the multiplication is defined by the rule $xa=\theta(a)x\, \, (a\in\mathbb{F}_q)$. A linear code $\C$ of length $n$ over $\mathbb{F}_q$ is said to be a skew cyclic code with respect to the automorphism $\theta$, if $(\theta(c_{n-1}),\theta(c_0),\ldots,\theta(c_{n-2}))\in \C$ for all $(c_0,c_1,\ldots,c_{n-1})\in \C$. There are many other studies defining skew cyclic codes over different rings, for example \cite{GR, FCR}.
In this paper, we are interested in studying skew cyclic codes over the ring $\ki$.

\section{Preliminary}

In this section  we  define first the structure of the Galois rings $\mathbb{Z}_{4}[\xi]$ and $\mathbb{Z}_{2}[\bar{\xi}]$ then the structure of  $\mathbb{Z}_{4}[\xi]$-module $\mathbb{Z}_{2}\mathbb{Z}_{4}[\xi]$.
Let $-: \mathbb{Z}_{4}\rightarrow \mathbb{Z}_{2}$ be the $\mod{2}$ reduction map.
This map is a homomorphism and can  naturally be extended to a homomorphism from $\mathbb{Z}_{4}[x]$ to $\mathbb{Z}_{2}[x]$.
Let h(x) be a monic polynomial over $\mathbb{Z}_{4}$. If $\overline{h}(x)$ is an irreducible polynomial over $\mathbb{Z}_{2}$, then $h(x)$ is called a monic basic irreducible polynomial. Moreover if  $\overline{h}(x)$ is a primitive polynomial, then $h(x)$ is called a monic basic primitive polynomial over $\mathbb{Z}_{4}$.


\begin{theorem}[\cite{wan}]
	Let $h(x)$ be a monic basic primitive polynomial of degree $m$ over $\mathbb{Z}_{4}$, then $\mathbb{Z}_{4}[x]/\langle h(x)\rangle$ is a Galois ring of characteristic $4$ and cardinality $4^m.$
\end{theorem}
Let $h(x)$ be a polynomial as in the Theorem above and write $\xi=x+\langle h(x)\rangle$ then $h(\xi)=0$, i.e. $\xi$ is a root of $h(x)\in \mathbb{Z}_{4}[x]$. Then each element of $\mathbb{Z}_{4}[x]/\langle h(x)\rangle$ can be uniquely expressed as
$$a_0+a_1\xi+\ldots+a_{m-1}\xi^{m-1}, a_i\in\mathbb{Z}_{4}, 0\leq i\leq m-1 .$$
Thus $\mathbb{Z}_{4}[x]/\langle h(x)\rangle=\mathbb{Z}_{4}[\xi]$, where $\mathbb{Z}_{4}[\xi]=\{a_0+a_1\xi+\ldots+a_{m-1}\xi^{m-1} | a_i\in\mathbb{Z}_{4}, 0\leq i\leq m-1 \}.$

Consider the following canonical homomorphism;
\begin{equation*}
\begin{split}
-:\mathbb{Z}_{4}[\xi]&\rightarrow \mathbb{Z}_{2}[\bar{\xi}]\\
\gamma=a_0+a_1\xi+\ldots+a_{m-1}\xi^{m-1}&\rightarrow \bar{\gamma}=\bar{a}_0+\bar{a}_1\bar{\xi}+\ldots+\bar{a}_{m-1}\bar{\xi}^{m-1}
\end{split}
\end{equation*}
where $\bar{\xi}$ is a root of the primitive polynomial $\overline{h}(x)\in \ZZ_2[x]$. Thus $\ZZ_{2}[\bar{\xi}]$ is a field extension of  $\ZZ_2$, i.e. $\ZZ_{2}[\bar{\xi}]$ is the field with $2^m$ elements, $\mathbb{F}_{2^m}$.
For further information on Galois rings readers may consult to \cite{wan}.


We define the set

\begin{equation*}
\mathbb{Z}_{2}\mathbb{Z}_{4}[\xi]=\left\{\left(\alpha,\beta\right) |~ \alpha\in\mathbb{Z}_{2}[\bar{\xi}]\text{ and }\beta\in \mathbb{Z}_{4}[\xi]\right\}.
\end{equation*}
Here, the sets $\mathbb{Z}_{2}[\bar{\xi}]=\{a_{0}+a_{1}\bar{\xi}+\cdots+a_{m-1}\bar{\xi}^{m-1}|a_{i}\in\ZZ_2\}$ and $\mathbb{Z}_{4}[\xi]=\{b_{0}+b_{1}\xi+\cdots+b_{m-1}\xi^{m-1}|b_{i}\in\ZZ_4\},~0\leq i\leq m-1$.

$\mathbb{Z}_{2}\mathbb{Z}_{4}[\xi]$ is an $\ZZ_4[\xi]$-module where the module multiplication is defined by $$\gamma*(\alpha,\beta)=(\bar{\gamma}\alpha,\gamma\beta) \text{ where } \gamma\in \ZZ_4[\xi] \text{ and } (\alpha,\beta)\in \ZZ_2[\bar{\xi}].$$
\begin{definition}
Let $\mathcal{C}$ be a non-empty subset of $\ki$. Then $\mathcal{C}$ is called a $\add[\xi]$-linear code if it is a $\ZZ_{4}[\xi]$-submodule of $\ki$.
\end{definition}

Since $\mathcal{C}$ is a $\ZZ_{4}[\xi]$-submodule of $\ki$ it is isomorphic to an abelian group of the form $\mathbb{Z}_{2}^{k_{0}}\times\mathbb{Z}_{4}^{k_{1}}$, where $k_0$ and $k_1$ are positive integers. Now consider the following sets.
\begin{equation*}
\mathcal{C}_{s}^{F}=\langle\{(a,b)\in
\ki~|~b~\text{free over}~\ZZ_{4}^{s}[\xi]\}\rangle~ \text{and}~dim(\mathcal{C}_{s}^{F})=k_{1}.
\end{equation*}
Let $D=\mathcal{C}
\backslash \mathcal{C}_{s}^{F}=\mathcal{C}_{0}\oplus \mathcal{C}_{1}$
such that
\begin{eqnarray*}
\mathcal{C}_{0} &=&\langle\{(a,b)\in
\ki~|~a\neq 0\}\rangle\subseteq \mathcal{C}\backslash \mathcal{C}_{s}^{F}\\
\mathcal{C}_{1} &=&\langle\{(a,b)\in\ki~|~a=0\}\rangle\subseteq \mathcal{C}\backslash \mathcal{C}_{s}^{F}.
\end{eqnarray*}
Hence, denote the dimension of $\mathcal{C}_{0}$ and $\mathcal{C}_{1}$ as a $k_{0}$ and $k_{2}$ respectively.
Considering all these parameters we say such a $\ZZ_{2}\ZZ_{4}[\xi]$-linear code $\mathcal{C}$ is of type $\left(r,s;k_0;k_1,k_2\right)$.

\bigskip

Next, for any elements
$$u=\left(a_{0},\ldots, a_{r-1},b_{0},\ldots,b_{s-1}\right), v=\left( d_{0},\ldots,d_{r-1},e_{0},\ldots,e_{s-1}\right)\in \ki,$$
define the inner product as
\begin{equation*}
\left\langle u,v\right\rangle =\left(2\sum_{i=0}^{r-1}a_{i}d_{i}+\sum_{j=0}^{s-1}b_{j}e_{j}\right)\in \ZZ_{4}[\xi].
\end{equation*}
According to this inner product, the dual linear code $\mathcal{C}^{\perp}$ of an any $\mathbb{Z}_{2}\mathbb{Z}_{4}[\xi]$-linear code $\mathcal{C}$ is defined in a usual way.
\begin{equation*}
\mathcal{C}^{\perp }=\left\{ v\in\ki|~\left\langle u,v\right\rangle =0~\text{for all}~u\in \mathcal{C}\right\} .
\end{equation*}
Therefore, if $\mathcal{C}$ is a $\mathbb{Z}_{2}\mathbb{Z}_{4}[\xi]$-linear code, then $\mathcal{C}^{\perp}$ is also a $\mathbb{Z}_{2}\mathbb{Z}_{4}[\xi]$-linear code.

\section{Generator and Parity-check Matrices of $\add[\xi]$-linear Codes}
In this section of the paper, we give standard forms of the generator and the parity-check matrices of a $\add[\xi]$-linear code $\C.$
\begin{theorem}
Let $\C$ be a $\add[\xi]$-linear code of type $\type$. Then $\C$ is permutation equivalent to a $\add[\xi]$-linear code which has the following standard form of the generator matrix,
\begin{eqnarray}\label{G}
G =
 \left(\begin{array}{cc|ccc}
  I_{k_0} & \bar{A}_{01} & 0 & 0 & 2T\\ \hline
  0 & S & I_{k_1} & A_{01} & A_{02} \\
   0 & 0& 0 & 2I_{k_2} & 2 A_{12}
 \end{array}\right)
\end{eqnarray}
where $\bar{A}_{01}$ is a matrix over $\ZZ_{2}[\bar{\xi}]$ and $A_{01}, A_{02}$ and $A_{12}$ are matrices with all entries from $\ZZ_{4}[\xi]$.
\end{theorem}
\begin{proof}
Let  $\C_{s}$ be the shortened code which only consists of the last $s$ coordinates of $\C$. Therefore, $\C_{s}$ is a linear code over $\ZZ_{4}[\xi]$ and has the following generator matrix of the form,
\begin{eqnarray*}
 \left(\begin{array}{cc|ccc}
   &  & I_{k_1} & A'_{01} & A'_{02}  \\
  & & 0 & 2I_{k} & 2A'_{12}
 \end{array}\right).
\end{eqnarray*}
Now adding the first $r$ coordinates to this generator matrix, we have
\begin{eqnarray*}
 \left(\begin{array}{cc|ccc}
   S_{1}&S_{2}  & I_{k_1} & A'_{01}&A'_{02} \\
  S_{3}&S_{4} & 0 & 2I_{k} & 2A'_{12}
 \end{array}\right)
 \end{eqnarray*}
 where $S_{i}$ are matrices with entries from $\ZZ_{2}[\xi]$ and $i\in\{1,2,3,4\}$. Next, by applying necessary row operations to last $k_2$ rows and row and column operations to the first $r$ coordinates we can rewrite this matrix as
 \begin{eqnarray*}
 \left(\begin{array}{cc|ccc}
   S'_{1}&S'_{2}  & I_{k_1} & A''_{01} & A''_{02}  \\
    0 & 0 & 0 &2I_{k_{2}} & 2A''_{12} \\
    I_{k_0}&S'_{4}& 0& 2I_{k'_{2}}&2A_{22}
 \end{array}\right).
 \end{eqnarray*}
Finally, by applying the necessary row and column operations to the above matrix, we have the result.
\end{proof}

\begin{example}
Let $\C$ be a $\add[\xi]$-linear code with $r=2, s=3$ and the following generator matrix. And let $\xi$ be the zero of an irreducible polynomial $x^2+x+1$ in $\ZZ_{2}[x]$ and $\ZZ_{4}[x]$.
\begin{eqnarray*}
 \left(\begin{array}{cc|ccc}
  1    & 1+\bar{\xi} & 2+2\xi& 2 & 2\\
  \bar{\xi}  & 0     & 2\xi  & 0 & 2 \\
   \bar{\xi} & 1     & 2+\xi & 1+3\xi & 0\\
   0   & 1+\bar{\xi} & 2\xi  & 2   & 1
 \end{array}\right)
\end{eqnarray*}
Hence, by applying necessary row operations we can write the standard form of this matrix as,
\begin{eqnarray}{\label{ex}}
 G=\left(\begin{array}{cc|ccc}
  1    & 0 & 0& 0 & 2\xi\\
  0  & 1     & 0  & 0 & 2+2\xi \\ \hline
   0 & 0     & 1 & 0 & 3\xi\\
   0   & 0 & 0  & 1   & 0
 \end{array}\right).
\end{eqnarray}
Looking at this standard form,
\begin{itemize}
\item
$\C$ is of type $(2,3;2;2,0)$
\item
$\C$ has $|C|=2^{2\cdot2}4^{2\cdot2}=4096$ codewords.

\end{itemize}
\end{example}

\begin{corollary}
If $\C$ is a $\add[\xi]$-linear code of type $\type$ where $\xi$ is a root of an irreducible polynomial of degree $t$, then $\C$ has
$$|\C|=2^{t(k_{0}+2k_{1}+k_{2})}$$
codewords.
\end{corollary}

\begin{theorem}
Let $\C$ be a $\add[\xi]$-linear code of type $\type$ with standard form of the generator matrix in (\ref{G}). Then $\C$ has parity-check matrix($\C^\perp$ has generator matrix) of the following standard form.
\begin{eqnarray*}
H=\left(
\begin{array}{ccccc}
-\bar{A}_{01}^{T} & I_{r-k_{0}} & -2S^{T} & 0 & 0 \\
-T^{T} & 0 & -A_{02}^{T}+A_{12}^{T}A_{01}^{T} & -A_{12}^{T} & I_{s-k_{1}-k_{2}} \\
0 & 0 & -2A_{01}^{T} & 2I_{k_{2}} & 0
\end{array}\right)
\end{eqnarray*}
\end{theorem}

\begin{proof}
It is easy to check that $G\cdot H^{T}=0$. So, every row of $H$ is orthogonal to the rows of $G$. Besides, by using the type of the matrix $H$, we have $|\C||\C^\perp|=2^{t\left(k_{0}+2k_{1}+k_{2}\right)}2^{t\left(r-k_{0}+2s-2k_{1}-2k_{2}+k_{2}\right)}=2^{t(r+2s)}$, where $t$ is the degree of an any irreducible polynomial which admits $\xi$ as a root. Therefore, we can conclude that the rows of $H$ not only are orthogonal to the rows of $G$ but also generate the all code $\C^{\perp}$.
\end{proof}

\begin{example}
Let $\C$ be a linear code over $\ZZ_{2}^2[\bar{\xi}]\times\ZZ_{4}^3[\xi]$ with the standard form of the generator matrix (\ref{ex}). Then,
\begin{eqnarray*}
H=\left(
\begin{array}{cc|ccc}
\bar{\xi} & 1+\bar{\xi} & \xi  & 1 &  0
\end{array}\right)
\end{eqnarray*}
is the generator matrix for the dual code $\C^\perp$. Further,
\begin{itemize}
\item
$\C^\perp$ is of type $\left(2,3;0;1,0\right)$
\item
$\C^\perp$ has $|\C^\perp|=4^{2\cdot1}=16$ codewords.
\end{itemize}
\end{example}
\begin{corollary}
If $\C$ is a $\add[\xi]$-linear code of type $\type$ then $\C^\perp$ is of type $\left(r,s;r-k_0;s-k_1-k_2,k_2\right)$.
\end{corollary}

\section{The Generators and the Spanning Sets for $\add[\xi]$-skew Cyclic Codes}

In this section, we introduce skew-cyclic codes over $\ki$. Since, $\ad$ is the extension of $\add$, the structure of linear and cyclic codes over $\ki$ is very similar to the structure of linear and cyclic codes over $\ZZ_{2}^r\times \ZZ_{4}^s$ and the structure of cyclic codes over $\ZZ_{2}\ZZ_{4}$ has been studied comprehensively in \cite{7}. Besides this similarity, we still give the definition and the structure of $\ad$-cyclic codes in order to discuss skew cyclic codes over $\ki$ more easily.
\begin{definition}
Let $\C$ be a linear code over $\ki$. $\C$ is called a $\ad$-cyclic code if for any codeword $c=\left( a_{0},a_{1},\ldots ,a_{r-1},b_{0},b_{1},\ldots, b_{s-1}\right) \in \mathcal{C},$ its cyclic shift
\begin{equation*}
\left( a_{r-1},a_{0},\ldots, a_{r-2}, b_{s-1},b_{0},\ldots,b_{s-2}\right)
\end{equation*}
is also in $\mathcal{C}$.
\end{definition}

\subsection{$\ad$-skew Cyclic Codes}

In this subsection we introduce two non-commutative rings $R_{i}[x,\theta_{i,t}]$ for $i=2,4$, where $R_{2}=\mathbb{Z}_{2}[\bar{\xi}]$, $R_{4}=\mathbb{Z}_{4}[\xi]$ and $\theta_i$ is the Frobenius automorphism of $R_i$. The structures of these rings are mainly related with the elements of finite ring $R_{i}$ and an automorphism $\theta_{i,t}$ of $R_{i}$.


The Frobenious automorphisms of $R_2$ and $R_4$ are defined as follows:
\begin{equation*} \label{theta}
\begin{split}
\theta_2  :R_{2}=\mathbb{Z}_{2}[\bar{\xi}]&\rightarrow R_{2}=\mathbb{Z}_{2}[\bar{\xi}] \\
 v_{0}+v_{1}\bar{\xi}+\ldots+v_{m-1}\bar{\xi}^{m-1}&\rightarrow v_{0}+v_{1}\bar{\xi}^2+\ldots+v_{m-1}\bar{\xi}^{2(m-1)}
\end{split}
\end{equation*}
\begin{equation*}
\begin{split}
\theta_4 :R_4=\mathbb{Z}_{4}[\xi]&\rightarrow R_4=\mathbb{Z}_{4}[\xi] \\
 v_{0}+v_{1}\xi+\ldots+v_{m-1}\xi^{m-1}&\rightarrow v_{0}+v_{1}\xi^2+\ldots+v_{m-1}\xi^{2(m-1)}.
\end{split}
\end{equation*}
Thus any automorphism of $R_i$ is a power of $\theta_i$. To avoid the abuse of notation we will denote the $t$-th power of $\theta_i$ with $\theta_{i,t}$.

\begin{definition}
The skew polynomial ring $R_{i}[x,\theta_{i,t}],~i\in \{2,4\}$ is a set of polynomials
\[
R_{i}[x,\theta_{i}]=\{a_{0}+a_{1}x+\cdots+a_{n}x^{n}|a_{j}\in R_{i},~j=0,1,\ldots,n \}
\]
where the multiplication $\ast$  is defined as
\[
(ax^{k})\ast (bx^{j})=a\theta_{i,t}^{k}(b)x^{k+j},~i\in \{2,4\}.
\]
and the addition of polynomials is defined as usual.
\end{definition}

Now, let us make a simple example to see that the skew polynomial rings are non-commutative.

\begin{example}
Consider the Galois ring $\mathbb{Z}_{4}[\xi]$, where $\xi$ is a root of the basic primitive polynomial $x^2+x+1$. And let $\theta_{4}$ be the automorphism defined by
\begin{eqnarray*}
\theta_{4}:&\mathbb{Z}_{4}[\xi]\longrightarrow \mathbb{Z}_{4}[\xi] \\
                     &v_{0}+v_{1}\xi\longrightarrow v_{0}+v_{1}\xi ^2
\end{eqnarray*}
Then,
\begin{eqnarray*}
(\xi x)\ast((1+\xi) x)=\xi \theta_{4}(1+\xi)x^{2}=\xi(1+\xi^{2})x^{2}=(1+\xi)x^{2}\\
((1+\xi) x)\ast(\xi x)=(1+\xi) \theta_{4}(\xi)x^{2}=(1+\xi)(\xi^{2})x^{2}=(3\xi)x^{2}.
\end{eqnarray*}
Hence, $(\xi x)\ast((1+\xi) x)\neq ((1+\xi) x)\ast(\xi x)$ in $R_4[x,\theta_4]$.

\end{example}

\begin{definition}
A  linear code $\C$  of length $n$ over $R_i$ is called a skew cyclic code if for any codeword $c=\left( c_{0},c_{1},\ldots ,c_{n-1}\right) \in \mathcal{C},$ its $\theta_{i,t}$-cyclic shift
\begin{equation*}
\sigma(c)=\left( \theta_{i,t}(c_{n-1}),\theta_{i,t}(c_{0}),\ldots, \theta_{i,t}(c_{n-2})\right)
\end{equation*}
is also in $\mathcal{C}$.
In polynomial representation, $c=\left( c_{0},c_{1},\ldots ,c_{n-1}\right) \in \mathcal{C}$  can be considered as a  polynomial of degree less than $n$ over $R_i[x,\theta_{i,t}]$, i.e. $c(x)=c_{0}+c_{1}x+\ldots+c_{n-1}x^{n-1}\in R_i[x,\theta_{i,t}]$. Thus $\sigma(c)$ corresponds to the polynomial $x\ast c(x) \in R_i[x,\theta_{i,t}]$. Therefore we can conclude that $\C$ is a skew cyclic code of length $n$ over $R_i$ if and only if $\C$ is a left $R_i[x,\theta_{i,t}]$-submodule of $R_i[x,\theta_{i,t}]/\langle x^n-1 \rangle$.
\end{definition}

 It is shown that  skew cyclic codes over finite fields are principally generated in \cite{D.B,any}.
Recall that $R_2=\ZZ_2[\bar{\xi}]$ is a field extension of $\ZZ_2$. Hence, if $\C$ is a skew cyclic code of length $n$ over $R_2$, then $\C=\langle f(x)\rangle$ where $f(x)$ is the unique monic polynomial of minimal degree in $\C$ (by Lemma 11 in \cite{any}). Moreover  $f(x)$ is a right divisor of $x^n-1$ in $R_2[x,\theta_{2,t}]$. We denote $g(x)$ is a right divisor of $h(x)$ by $g(x)|_rh(x)$ through the paper.

The following theorem can be proven using  similar methods as in the proof the theorem of skew cyclic codes over $\mathbb{F}_{p^m}+u\mathbb{F}_{p^m}$ where $u^2=0$ in \cite{FCR}.

\begin{theorem}\label{Z4hali}
Let $\C$ be a skew cyclic code of length $s$ over $\ZZ_{4}[\xi]$ and $A$ be the set of minimal degree polynomials of $C$.
	\begin{enumerate}[$i)$]
		\item If there exist no monic skew polynomials in $\C$, then $\C=\langle2q(x)\rangle$ where $q(x)|_rx^s-1\left(mod~2 \right)$.
		\item\label{iki} If there exists a monic skew polynomial in $A$, then $\C=\langle g(x)+2a(x)\rangle$ where $g(x)+2a(x)|_rx^s-1$. Moreover   $g(x)|_rx^s-1\left(mod~2 \right)$ and $deg(a(x))<deg(g(x))$.
		\item If there exist no monic skew polynomials in $A$ but there exists a monic skew polynomial in $\C$ then, $\C=\langle g(x)+2a(x), 2q(x)\rangle$ where $deg(a(x))<deg(q(x))<deg(g(x))<s$, $q(x)|_rg(x)|_rx^s-1\left(mod~2 \right)$ and $x^s-1=h_{g}(x)g(x)$ and $q(x)|_rh_{g}(x)a(x)(mod~2)$.
	\end{enumerate}
\end{theorem}

Now, let $\C$ be a $\ad$-linear code and let $c=\left( a_{0}a_{1}\ldots a_{r-1},b_{0}b_{1}\ldots b_{s-1}\right) \in\ki$ be an any codeword in $\C$. Then, $c$ can be identified with a module element  consisting of two polynomials each from different rings such that
\begin{eqnarray*}
	c(x) &=&\left( a_{0}+a_{1}x+\ldots +a_{r-1}x^{r-1},b_{0}+b_{1}x+\ldots +b_{s-1}x^{s-1}\right) \\
	&=&\left( a(x),b(x)\right)
\end{eqnarray*}
in  $R_{\theta}=R_{2}[x,\theta_{2}]/\langle x^{r}-1\rangle \times R_{4}[x,\theta_{4}]/\langle x^{s}-1\rangle.$
This identification gives a one-to-one correspondence between elements in
$\ki$ and elements in $R_{\theta}$.

Let $f(x)=f_{0}+f_{1}x+\ldots +f_{k}x^{k}\in R_{4}[x,\theta_{4}]$, $(g(x),h(x))\in R_{\theta}$ and consider the following multiplication
\begin{equation*}
f(x)\ast\left( g(x),h(x)\right) =\left(f(x)\ast g(x)~mod~2,f(x)\ast h(x)~mod~4\right).
\end{equation*}
It is obvious that this multiplication is well-defined and $R_\theta$ is a left $R_{4}[x,\theta_{4}]$-module with respect to this multiplication.

\begin{definition}
	Let $\C$ be a linear code over $\ki$. $\C$ is called a $\ad$-skew cyclic code if for any codeword $c=\left( a_{0},a_{1},\ldots ,a_{r-1},b_{0},b_{1},\ldots, b_{s-1}\right) \in \mathcal{C},$ its $\theta$-cyclic shift
	\begin{equation*}
	\theta(c)=\left( \theta_{2}(a_{r-1}),\theta_{2}(a_{0}),\ldots, \theta_{2}(a_{r-2}), \theta_{4}(b_{s-1}),\theta_{4}(b_{0}),\ldots,\theta_{4}(b_{s-2})\right)
	\end{equation*}
	is also in $\mathcal{C}$.
	
\end{definition}

\begin{theorem}
	Let $\C$ be a $\ad$-linear code. $\C$ is $\ad$-skew cyclic code if and only if $\C$ is a left $R_{4}[x,\theta_{4}]$-submodule of $R_\theta$.
\end{theorem}

From above discussion,  if $\C$ is a $\ad$-skew cyclic code, we have
\begin{eqnarray*}
\Psi:\C\rightarrow R_{4}[x,\theta_{4}]/\langle x^{s}-1\rangle \\
(f_{1}(x),f_{2}(x))\rightarrow f_{2}(x)
\end{eqnarray*}
which is a  left $R_{4}[x,\theta_{4}]$-module homomorphism and its image$(\Psi )$ is a
left $R_{4}[x,\theta_{4}]$-submodule  of $R_{4}[x,\theta_{4}]/\langle x^{s}-1\rangle$ and $\ker (\Psi )$
is a submodule of $\mathcal{C}$.

\begin{theorem}\label{theo1}
Let $\C$ be a $\ad$-skew cyclic code. Then $\C$ can be identified as
\begin{enumerate}[$i)$]

\item If $\Psi(\mathcal{C})=\langle 2q(x)\rangle$,  then $\mathcal{C}=\langle \left( f(x),0\right) ,\left(l_{1}(x),2q(x)\right) \rangle$ where $q(x)|_rx^s-1 \left(mod~2 \right)$ and $l(x)$
is a binary polynomial satisfying $f(x)|rh_{q}(x)l_{1}(x)\left(mod~2 \right)$ where $x^s-1= h_{q}( x) q(x)\left(mod~2 \right)$.

\item If $\Psi(\C)=\langle g(x)+2a(x)\rangle$, then $\C=\langle \left( f(x),0\right) ,\left(
l(x),g(x)+2a(x)\right) \rangle$ where $f(x)|_r\left( x^{r}-1\right)(mod~2),$ and $g(x)+2a(x)|_r\left( x^{s}-1\right) ,$ and $l(x)$
is a binary polynomial satisfying $\deg (l(x))<\deg (f(x)),$ $f(x)|_r h_{g,a}( x) l(x)\left(mod~2 \right)$ where $x^s-1= h_{g,a}( x) (g(x)+2a(x))$.

\item  If $\Psi(\C)=\langle g(x)+2a(x), 2q(x)\rangle$, then
$\mathcal{C}=\langle \left( f(x),0\right) ,\left(
l(x),g(x)+2a(x)\right),\left(l_{1}(x),2q(x)\right) \rangle$  where $f(x)|_r\left( x^{r}-1\right)(mod~2),$ and $q(x)|_rg(x)|_r\left( x^{s}-1\right) \left(mod~2\right),$ and $q(x)|_rh_{g}(x)a(x)(mod~2)$ and
 $l(x),~l_{1}(x)$ are binary polynomials of degree less than the degree of $f(x)$ with $f(x)|_rh_{q}(x)l_{1}(x)\left(mod~2 \right)$ and $f(x)|_rk(x)l_{1}(x)+h_{g}(x)l(x)\left(mod~2 \right)$ where $k(x)q(x)=h_{g}(x)a(x)(mod~2)$ and $q(x)h_{q}(x)=x^{s}-1$ with $\deg(a(x))<\deg(q(x))$.

\end{enumerate}
\end{theorem}
\begin{proof}
	We will prove only $ii$).  The proof of  $i$) and $iii$) are similar to the proof of $ii$).\\
Suppose that
 $\Psi\left(\C\right)=\langle g(x)+2a(x)\rangle$ as in Theorem \ref{Z4hali} $ii$).   Further,
\begin{eqnarray*}
\ker (\Psi )=\{\left(f(x),0\right)\in \C: f(x)\in R_{2}[x,\theta_{2}]/(x^{r}-1)\}.
\end{eqnarray*}

Let $I=\{f(x)\in R_{2}[x,\theta_{2}]|(f(x),0)\in \ker (\Psi )\}$. It can be easily seen that $I$ is a submodule of the ring $R_{2}[x,\theta_{2}]$, therefore principally generated and we may assume that $I=\langle f(x)\rangle$.
 Since for any $(w(x),0)\in ker(\Psi)$, $w(x)\in I$ and $w(x)=w_1(x)f(x)$ for some $w_1(x)\in R_{2}[x,\theta_{2}]$, then $(w(x),0)\in \langle (f(x),0)\rangle$
and therefore, $ker(\Psi )$ is a left-submodule of $\C$ generated by $\left(f(x),0\right)$ with $f(x)|_{r}(x^r-1)(mod~2)$.
 By using the first isomorphism theorem we have,
\[
\C/ker(\Psi)\cong \langle g(x)+2a(x) \rangle.
\]
Now, let $\left( l(x),g(x)+2a(x) \right)\in \C$ such that
\[
\Psi\left(l(x),g(x)+2a(x)\right)=\langle g(x)+2a(x)\rangle.
\]

Hence we can conclude that  $\C$ can be generated as a left $R_{4}[x,\theta_{4}]$-submodule  of $R_{\theta}$ by two elements of the form $(f(x),0)$ and $(l(x),g(x)+2a(x))$. So, any element in $\C$ can be expressed as
\[
d(x)\ast\left(f(x),0\right)+e(x)\ast\left(g(x)+2a(x)\right),
\]
where $d(x),e(x)\in R_{4}[x,\theta_{4}]$. We can restrict the polynomial $d(x)$ over $R_{2}[x,\theta_{2}]$ and therefore we can write,
\[
\C=\langle\left(f(x),0\right),\left(l(x),g(x)+2a(x)\right)\rangle.
\]


\end{proof}

We know that if $\mathcal{C}$ is a $\ad$-skew cyclic code of the form
\newline
$\mathcal{C}=\langle \left( f(x),0\right) ,\left(l(x),g(x)+2a(x)\right),\left(l_{1}(x),2q(x)\right) \rangle$  with $g(x)\neq 0$, then $\mathcal{C}$ is a free $R_{4}$-module. If $\mathcal{C}$ is
not of this form then it is not  a free $R_{4}$-module. But we still present a minimal spanning set for the code. The following theorem gives us a spanning minimal set for $\ad$-skew cyclic codes.

\begin{theorem}\label{theo2}
Let $\mathcal{C}=\langle \left( f(x),0\right) ,\left(l(x),g(x)+2a(x)\right),\left(l_{1}(x),2q(x)\right) \rangle$ be a $\ad$-skew cyclic code in $R_{\theta}$ where $f(x),~l(x),~l_{1}(x),~g(x)$ and $q(x),~a(x)$ are as in Theorem \ref{theo1} and $f(x)h_{f}(x)=x^{r}-1$.
Let
\begin{equation*}
S_1=\bigcup_{i=0}^{deg(h_{f})-1}\left\{x^i \ast \left( f(x),0\right)\right\},
\end{equation*}

\begin{equation*}
S_{2}=\bigcup_{i=0}^{deg(h_{g})-1}\left\{x^{i}\ast \left(l(x),g(x)+2a(x)\right) \right\}
\end{equation*}

\begin{equation*}
S_{3}=\bigcup_{i=0}^{deg(h_{q})-1}\left\{ x^{i}\ast \left(l_{1}(x),2q(x)\right) \right\} .
\end{equation*}


Then
\begin{equation*}
S=S_{1}\cup S_{2}\cup S_{3}
\end{equation*}
forms a minimal spanning set for $\mathcal{C}$ and $\mathcal{C}$ has $2^{deg(\xi){deg(h_{f})}}4^{deg(\xi){deg(h_{g})}}2^{deg(\xi){deg(h_{q})}}$ codewords.
\end{theorem}

\begin{proof}
Let, $c(x)\in R_{\theta}$ be a codeword in $\C$. Therefore, we can write $c(x)$ as,
\[
c(x)=\bar{d}(x)\ast\left(f(x),0\right)+e_{1}(x)\ast\left(l(x),g(x)+2a(x)\right)+e_{2}(x)\ast\left(l_{1}(x),2q(x)\right)
\]
where $d(x),e_{1}(x)$ and $e_{2}(x)$ are polynomials in $R_{4}[x,\theta_{4}]$.
Now, if $\deg\bar{d}(x)\leq \deg(h_{f}(x))-1$ then $\bar{d}(x)\ast\left(f(x),0\right)\in Span(S_{1})$. Otherwise, by using right division algorithm we have

\[
\bar{d}(x)=h_{f}(x)\bar{q_{1}}(x)+\bar{r_{1}}(x)
\]
where $\bar{r_{1}}(x)=0$ or $\deg(\bar{r_{1}}(x))\leq \deg(h_{f}(x))-1$. Therefore,
\begin{eqnarray*}
\bar{d}(x)\ast \left(f(x),0\right)&=&\left(h_{f}(x)\bar{q_{1}}(x)+\bar{r_{1}}(x)\right)\ast \left(f(x),0\right)\\
&=&\bar{r_{1}}(x)\ast\left(f(x),0\right).
\end{eqnarray*}
Hence, we can assume that $\bar{d}(x)\ast\left(f(x),0\right)\in Span(S_{1})$.

If $\deg(e_{2}(x))\leq \deg(h_{q}(x))-1$ then $e_{2}(x)\ast\left(l_{1}(x),2q(x)\right)\in Span(S_{3})$, otherwise using the right division algorithm, we have polynomials $q_{4}(x)$ and $r_{4}(x)$ such that
\[
e_{2}(x)=q_{4}(x)h_{q}(x)+r_{4}(x)
\]
where $r_{4}(x)=0$ or $\deg(r_{4}(x))\leq \deg(h_{q}(x))-1$.
Therefore,
\begin{eqnarray*}
	e_{2}(x)\ast\left(l_{1}(x),2q(x)\right)&=&\left(q_{4}(x)h_{q}(x)+r_{4}(x)\right)\ast\left(l_{1}(x),2q(x)\right)\\
	&=&q_{4}(x)\ast\left(h_{q}(x)l_{1}(x),0\right)+r_{4}(x)\ast\left(l_{1}(x),2q(x)\right).
\end{eqnarray*}
Since $f(x)|_{r}h_{q}(x)l_{1}(x)$ we have $q_{4}(x)\ast\left(h_{q}(x)l_{1}(x),0\right)\in Span(S_{1})$. Also  $r_{4}(x)\ast\left(l_{1}(x),2q(x)\right)\in Span(S_{3})$, thus $e_{2}(x)\ast\left(l_{1}(x),2q(x)\right)\in Span(S_1\cup S_{3})$.

Now, if $\deg(e_{1}(x))\leq \deg(h_{g}(x))-1$ then $e_{1}(x)\ast\left(l(x),g(x)+2a(x)\right)\in Span(S_{2})$, otherwise again by the right division algorithm, we get polynomials $q_{2}(x)$ and $r_{2}(x)$ such that
\[
e_{1}(x)=q_{2}(x)h_{g}(x)+r_{2}(x)
\]
where $r_{2}(x)=0$ or $\deg(r_{2}(x))\leq \deg(h_{g}(x))-1$.
So, we have
\begin{eqnarray*}
e_{1}(x)\ast\left(l(x),g(x)+2a(x)\right)&=&\left(q_{2}(x)h_{g}(x)+r_{2}(x)\right)\ast\left(l(x),g(x)+2a(x)\right)\\
&=&q_{2}(x)\ast\left(h_{g}(x)l(x),2h_{g}(x)a(x)\right)+r_{2}(x)\ast\left(l(x),g(x)+2a(x)\right).
\end{eqnarray*}
Since $r_{2}(x)=0$ or $\deg(r_{2}(x))\leq \deg(h_{g}(x))-1$, $r_{2}(x)\ast\left(l(x),g(x)+2a(x)\right)\in Span(S_{2})$. Let us  consider $q_{2}(x)\ast\left(h_{g}(x)l(x),2h_{g}(x)a(x)\right)$.

From Theorem \ref{theo1}; $f(x)|_rk(x)l_{1}(x)+h_{g}(x)l(x)\left(mod~2 \right)$ so, there exist $\mu(x)\in R_2[x,\theta_2]$ such that $\mu(x)f(x)=k(x)l_{1}(x)+h_{g}(x)l(x)$. Then $h_{g}(x)l(x)=\mu(x)f(x)+k(x)l_{1}(x).$ Also $h_g(x)a(x)=k(x)q(x).$
Thus,
\begin{align*}
	q_{2}(x)\ast\left(h_{g}(x)l(x),2h_{g}(x)a(x)\right)&=q_{2}(x)\ast\left[
	k(x)\ast\left(l_1(x),q(x)\right)+\mu(x)\ast\left(f(x),0\right)\right]\\
	&=q_2(x)k(x)\left(l_1(x),q(x)\right)+q_2(x)\mu(x)\left(f(x),0\right).
\end{align*}

Since $q_2(x)k(x)\left(l_1(x),q(x)\right)\in Span(S_1\cup S_{3})$  and $q_2(x)\mu(x)\left(f(x),0\right)\in Span(S_1)$, then we have $e_{1}(x)\ast\left(l(x),g(x)+2a(x)\right)\in Span(S)$.

 Consequently, $S_{1}\cup S_{2}\cup S_{3}$ forms a minimal spanning set for $\mathcal{C}$.
\end{proof}

\begin{example}
Let $\mathcal{C}$ be a $\ad$-cyclic code in $R_{2}[x,\theta_{2}]/(x^{7}-1)\times R_{4}[x,\theta_{4}]/(x^{7}-1)$ and $\xi$ be the zero of an irreducible polynomial $x^2+x+1$ in $\ZZ_{2}[x]$ and $\ZZ_{4}[x]$. Consider that $\mathcal{C}$ is generated by $\left((f(x),0),(l(x),g(x)+2a(x))\right) $ where
\begin{eqnarray*}
f(x) &=&1+ x + x^3, \\
g(x) &=&1 + 2 x + 3 x^2 + x^3 + x^4, \\
a(x) &=&3 + x, \\
l(x) &=&1+ x^2
\end{eqnarray*}
and also
\begin{eqnarray*}
x^7-1 &=&(1+x) \left(1+x+x^3\right) \left(1+x^2+x^3\right)~\text{in}~R_{2}[x,\theta_2]  \\
x^7-1 &=&(3+x) \left(3+x+2x^2+x^3\right)(3+2x+3x^2+x^3)~\text{in}~R_{4}[x,\theta_4]
\end{eqnarray*}
\bigskip

Therefore, we can calculate the following polynomials.
\begin{eqnarray*}
f(x)h_{f}(x) &=&x^{7}-1\Rightarrow h_{f}(x)=1 + x + x^2 + x^4, \\
g(x)h_{g}(x) &=&x^{7}-1\Rightarrow h_{g}(x)=3 + 2 x + 3 x^2 + x^3. \\
\end{eqnarray*}
Hence, using the generator sets in Theorem \ref{theo2}, we can write the generator matrix for a $\ad$-skew cyclic code $\mathcal{C}$ as follows.
$$G=\left(
\begin{array}{cccccccccccccc}
 1 & 1 & 0 & 1 & 0 & 0 & 0 & 0 & 0 & 0 & 0 & 0 & 0 & 0 \\
 0 & 1 & 1 & 0 & 1 & 0 & 0 & 0 & 0 & 0 & 0 & 0 & 0 & 0 \\
 0 & 0 & 1 & 1 & 0 & 1 & 0 & 0 & 0 & 0 & 0 & 0 & 0 & 0 \\
 0 & 0 & 0 & 1 & 1 & 0 & 1 & 0 & 0 & 0 & 0 & 0 & 0 & 0 \\
 1 & 0 & 1 & 0 & 0 & 0 & 0 & 3 & 0 & 3 & 1 & 1 & 0 & 0 \\
 0 & 1 & 0 & 1 & 0 & 0 & 0 & 0 & 3 & 0 & 3 & 1 & 1 & 0 \\
 0 & 0 & 1 & 0 & 1 & 0 & 0 & 0 & 0 & 3 & 0 & 3 & 1 & 1 \\
 1 & 0 & 0 & 1 & 1 & 1 & 0 & 2 & 2 & 2 & 0 & 2 & 0 & 0 \\
 0 & 1 & 0 & 0 & 1 & 1 & 1 & 0 & 2 & 2 & 2 & 0 & 2 & 0 \\
 1 & 0 & 1 & 0 & 0 & 1 & 1 & 0 & 0 & 2 & 2 & 2 & 0 & 2
\end{array}\right).$$
Moreover, $\mathcal{C}$ is of type $(7,7;4,3,3)$.
\end{example}

\begin{example}
Let $r=s=4$ and $\xi$ be a root of the basic primitive polynomial $x^2+x+1$. One of the factorizations of $x^4-1$ is
\begin{eqnarray*}
x^4-1=\left(x^2+\bar{\xi}^{2}x+\bar{\xi}^{2}\right)\left(x^2+\bar{\xi}^{2}x+\bar{\xi}\right)~\text{in}~R_{2}[x,\theta_2]~\text{and}\\
x^4-1=\left(x^2+2{\xi}x+3\right)\left(x^2+2{\xi}x+1\right)~\text{in}~R_{4}[x,\theta_4].
\end{eqnarray*}
Then, let $\C$ be a skew cyclic code generated by $\langle \left( f(x),0\right) ,\left(l(x),g(x)+2a(x)\right),\left(l_{1}(x),2q(x)\right) \rangle$ where

\begin{eqnarray*}
f(x)=x^2+\bar{\xi}^2x+\bar{\xi},~l(x)=1,~l_{1}(x)=\bar{\xi}x+\bar{\xi}\\
g(x)=1+x^2=q(x),~a(x)={\xi}.
\end{eqnarray*}
Therefore, we can calculate the following polynomials:

\begin{eqnarray*}
k(x)q(x)=h_{g}(x)a(x)\Longrightarrow k(x)={\xi}\\
q(x)h_{q}(x)=x^4-1\Longrightarrow h_{q}(x)=x^2-1.
\end{eqnarray*}
Finally, with the help of Theorem \ref{theo2}, we have the generator matrix of $\C$ as,
$$G=\left(
\begin{array}{cccccccc}
 \bar{\xi} & \bar{\xi}^{2} & 1 & 0 & 0 & 0 & 0 & 0  \\
 0 & \bar{\xi}^{2} & \bar{\xi} & 1 & 0 & 0 & 0 & 0  \\
 1 & 0 & 0 & 0 & 1+2\xi & 0 & 1 & 0  \\
 0 & 1 & 0 & 0 & 0 & 1+2\xi^2 & 0 & 1  \\
 \bar{\xi} & \bar{\xi} & 0 & 0 & 2 & 0 & 2 & 0  \\
 0 & \bar{\xi}^{2} & \bar{\xi}^{2} & 0 & 0 & 2 & 0 & 2
\end{array}\right).$$

\end{example}


\begin{thebibliography}{10}

\bibitem {7} Abualrub, T., Siap, I. and Aydin, N.,\textquotedblleft $\mathbb{Z}_{2}\mathbb{Z}_{4}$-additive cyclic codes\textquotedblright, IEEE Trans. Info.
Theory, vol. 60, No. 3, pp. 1508-1514, Mar. 2014.



\bibitem{1} Aydogdu, I. and Siap, I., \textquotedblleft The Structure of $\mathbb{Z}_{2}\mathbb{Z}_{2^{s}}-$Additive Codes: Bounds on the minimum distance\textquotedblright, Applied Mathematics and Information Sciences(AMIS),7, (6), 2271-2278 2013.

\bibitem {6} Aydogdu, I. and Siap, I., \textquotedblleft On $\mathbb{Z}_{p^{r}}\mathbb{Z}_{p^{s}}-$additive codes\textquotedblright, Linear and Multilinear Algebra, vol. 63, no. 10, pp. 2089-2102, 2015.


\bibitem{2} Aydogdu, I., Abualrub, T. and Siap, I., \textquotedblleft On $\mathbb{Z}_{2}\mathbb{Z}_{2}[u]-$additive codes\textquotedblright , International Journal of Computer Mathematics, vol.92, no. 9, pp. 1806-1814, 2015.


\bibitem{3} Borges, J., Fern\'{a}ndez-C\'{o}rdoba, C., Pujol, J., Rif\`{a}, J. and Villanueva, M., \textquotedblleft $\mathbb{Z}_{2}\mathbb{Z}_{4}$-linear codes: Generator Matrices and Duality\textquotedblright, Designs, Codes and Cryptography, 54, (2), 167-179, 2010.
\bibitem {D.B} Boucher, D., Geiselmann, W. and Ulmer, F., ``Skew cyclic codes", Appl. Algebr. Eng. Comm., vol. 18, pp. 379--389, 2007.
\bibitem {GR}Boucher, D., Sol\'e, P. and Ulmer, F., ``Skew constacyclic codes over Galois rings", Adv. Math. Commun., vol. 2, no. 3, pp. 273--292, 2008.

\bibitem {any}I. Siap, T. Abualrub, N. Aydin and P. Seneviratne, Skew cyclic codes of arbitrary length, Int.J. Inform. Coding Theory, 2 (2011), 10--20.

\bibitem {FCR}Jitman, S., Ling, S. and Udomkavanich, P., ``Skew constacyclic codes over finite chain rings",  Adv. Math. Commun., vol. 6, no. 1, 39--63, 2012.
	
\bibitem{wan} Wan, Z. X., Lectures on Finite Fields and Galois Rings,  World Scientific, Singapore (2003).
\end{thebibliography}
\end{document}